\documentclass[conference,letterpaper]{IEEEtran}
\usepackage{cite}
\usepackage{amsmath,amssymb,amsfonts}
\usepackage{graphicx}
\usepackage{textcomp}
\usepackage{xcolor}
\def\BibTeX{{\rm B\kern-.05em{\sc i\kern-.025em b}\kern-.08em
    T\kern-.1667em\lower.7ex\hbox{E}\kern-.125emX}}

\usepackage[utf8]{inputenc} 
\usepackage[T1]{fontenc}
\usepackage{url}
\usepackage{ifthen}
\usepackage{latexsym}
\usepackage{amsthm}
\usepackage{caption}
\usepackage{amssymb}
\usepackage{mdframed}
\usepackage[colorlinks=true, allcolors=blue] {hyperref}

\usepackage{mathtools}
\usepackage{braket}

\usepackage{algorithm,algorithmic}

\makeatletter
\def\th@theoremstyle{\thm@headfont{\bfseries}}
\makeatother
\theoremstyle{theoremstyle}  
\newtheorem{theorem}{Theorem} 
\newtheorem{lemma}[theorem]{Lemma} 
\newtheorem{definition}[theorem]{Definition}

\begin{document}

\title{Classical-Quantum Channel Resolvability Using\\ Matrix Multiplicative Weight Update Algorithm} 


\author{\IEEEauthorblockN{Koki Takahashi}
\IEEEauthorblockA{Department of Electrical Engineering and Computer Science \\
\textit{Tokyo University of Agriculture and Technology}\\
Tokyo, Japan \\
s251249z@st.go.tuat.ac.jp}
\and
\IEEEauthorblockN{Shun Watanabe}
\IEEEauthorblockA{Department of Electrical Engineering and Computer Science \\
\textit{Tokyo University of Agriculture and Technology}\\
Tokyo, Japan \\
shunwata@cc.tuat.ac.jp}
}

\maketitle

\begin{abstract}
We study classical-quantum (C-Q) channel resolvability. C-Q channel resolvability has been proved by only random coding in the literature. In our previous study, we proved channel resolvability by deterministic coding, using multiplicative weight update algorithm. We extend this approach to C-Q channels and prove C-Q channel resolvability by deterministic coding, using the matrix multiplicative weight update algorithm. This is the first approach to C-Q channel resolvability using deterministic coding.
\end{abstract}






\section{Introduction}
Shannon introduced random coding to prove channel coding \cite{Shannon48}. Random coding is the most widely used approach in information theory. Feinstein also proved channel coding by maximal code construction \cite{Feinstein59}. 
Feinstein's deterministic coding approach has intensified our understanding of the channel coding as seen in \cite{Ash65, Blackwell_Breiman_Thomasian59, Csiszár96, Han03_Information-Spectrum, Ogawa_Nagaoka07, Winter99}.

Han and Verdu introduced channel resolvability \cite{Han_Verdu93}, which plays a key role in proving the strong converse for identification via channels \cite{Alswede_Dueck89}.  Ahlswede 
also presented an alternative proof of its strong converse later \cite{Ahlswede06}.

Ahlswede and Winter proved strong converse for identification via classical-quantum (C-Q) channels \cite{Ahlswede_Winter02} by implicitly proving the achievability of C-Q channel resolvability.  The result in \cite{Ahlswede_Winter02} has been used to prove the wiretap capacity of quantum channels \cite{Devetak_05, Cai_Winter_Yeung_04} and the quantum capacity of quantum channels \cite{Devetak_05}.
Thus, C-Q channel resolvability can play an important role in studying quantum information theory. 



To the best of our knowledge, C-Q channel resolvability has been constructed by only random coding.   
Our previous study showed that, for channel resolvability of classical channels, deterministic coding constructed by the multiplicative weight update (MWU) algorithm achieves the channel capacity \cite{takahashi2025channelresolvabilityusingmultiplicative}. In this paper, we extend the result in \cite{takahashi2025channelresolvabilityusingmultiplicative} to C-Q channels.
More specifically, we show that, for C-Q channel resolvability, deterministic coding constructed by the matrix MWU (MMWU) algorithm achieves the Holevo capacity.

The MMWU algorithm is a widely used algorithm for game theory and learning theory in quantum computing, and it approximately finds a certain minmax value of a two-player game \cite{Arora_Hazan_Kale12}. By interpreting C-Q channel resolvability as a two-player game, we apply the MMWU algorithm. The key point is that Ahlswede and Winter used quantum hypergraph soft covering to prove C-Q channel resolvability. 
Instead of the quantum Chernoff bound used in \cite{Ahlswede_Winter02}, we apply the MMWU algorithm to solve quantum hypergraph soft covering.
Our proof is inspired by Kale's thesis that the MMWU algorithm solves quantum hypergraph covering \cite{Kale07}; however, there are technical difficulties inherent in soft covering.

In recent study, refined analysis of C-Q channel resolvability and soft covering for random coding has been studied in \cite{Shen_Gao_Cheng24, Cheng_Gao24, He_Atif_Pradhan_2024, Hayashi_Cheng_Gao_25, Anshu_Hayashi_Warsi_20,Cheng_Gao_Hirche_Huang_Liu_2025, Atif_Pradhan_Winter_24, Matsuura_Hayashi_Hsieh_2025}. In addition, universal C-Q channel resolvability has been proposed in \cite{Matsuura_Hayashi_Hsieh_2025}; however the expander code used in \cite{Matsuura_Hayashi_Hsieh_2025} is constructed by random coding. 

\textit{Notation:}
Throughout the paper, $e$ is Napier's constant and the base of the natural logarithm is $e$. 
Finite alphabet sets (e.g. $\mathcal{X}$),
Hilbert space (e.g. $\mathcal{H}$) and set of density matrices of given Hilbert space (e.g. $\mathcal{S(H)}$) are denoted by calligraphic letters. 
The cardinality and dimension are denoted by $\mathcal{|X|}, \dim \mathcal{H}$ respectively. The set of all distributions on $\mathcal{X}$ is denoted by $\mathcal{P(X)}$. For a Hermitian matrix $A$, the generalize inverse $A^{-1}$ is defined such that the eigenvalues are inverse of nonzero eigenvalues of $A$; see also \cite[Definition 3.3.1]{Wilde_2018}. 
\section{Problem formulation of C-Q channel resolvability }

Let $W: \mathcal{X} \rightarrow \mathcal{S(H)}$ be a C-Q channel that maps each input $x$ to a quantum state $W_x$ on the output Hilbert space $\mathcal{H}$. Given an input distribution $p \in {\cal P(X)},$ the goal is to simulate the output quantum state
$$W_p := \sum_{{x}\in{\cal X}}p(x) W_{x}.$$ 
Let ${\cal C} = \{{x}_1, {x}_2, ...,x_{L}\}$ be a codebook of size $L$. Define $W_{{\cal C}}$ as the output mixed state induced by selecting a codeword uniformly at random from codebook ${\cal C}$ and passing the codeword through the C-Q channel, i.e.,
$$W_{{\cal C}}:=\frac{1}{L}\sum_{l=1}^L W_{x_l}.$$
The problem is how to construct a codebook ${\cal C}$ such that the output state $W_\mathcal{C}$ simulates the output state $W_p$.
A simulation error is measured by the trace distance: 
$$d_\mathrm{tr} (W_p, W_{{\cal C}}) := \frac{1}{2}\|W_p -W_{{\cal C}}\|_1,$$
where $\|A\|_1$ is trace norm of matrix $A$.
We shall make a codebook ${\cal C}$ such that, for an arbitrary small $\xi$, the trace distance $d_\mathrm{tr}(W_{p}, W_{{\cal C}})$ is less than $\xi,$ with codebook size $L$ as small as possible.

For the asymptotic setting of discrete memoryless C-Q channel $W$ and the worst input distribution, the optimal rate for C-Q channel resolvability is given by the Holevo capacity $C(W)$. Our paper achieves this result using deterministic coding in Theorem \ref{the:quantum_channel_resolvability_}.  


\section{Matrix multiplicative weight algorithm}

In this section, we review the MMWU algorithm in notations that are compatible with
C-Q channel resolvability.
We consider an $L$ rounds game between ${\cal X}$-player and ${\cal H}$-player with a quantum system;
${\cal X}$-player determines the measurement "cost matrix" $0\preceq M(x)\preceq I$ indexed by  $x \in {\cal X}$ and ${\cal H}$-player inputs state $\ket{\psi}\in {\cal H}$ into the measuring system (mixed state is also possible).
After their action, the cost $\bra{\psi}M(x)\ket{\psi}$ is incurred.
Consider ${\cal X}$-player seeks to maximize the cost while ${\cal H}$-player seeks to minimize the cost.
In each round $1 \le l \le L$, ${\cal H}$-player inputs a (mixed) state, i.e., 
a density matrix $F(l)$ on $\mathcal{S(H)}$; then ${\cal X}$-player responds cost matrix $M(x_l)$ with pure
strategy $x_l \in {\cal X}$. Our goal is to construct a sequence of mixed states
$F(1),F(2), \ldots, F(L)$ such that the total cost $\sum_{l=1}^L \Big(\mathrm{Tr}\Big( F(l) M(x_l) \Big)\Big)$
is not too much more than the cost of ${\cal H}$-player's best pure state $\ket{\psi}\in\mathcal{H}$ in hindsight, i.e., $\min_{\ket{\psi}\in {\cal H}} \bra{\psi} \sum_{l=1}^L  M(x_l) \ket{\psi}.$

The MMWU algorithm is an iterative algorithm that approximately obtains each player's best mixed strategies. The MMWU algorithm has a few variations, and we use so-called Hedge algorithm \cite{Arora_Hazan_Kale12}. 
The MMWU algorithm reviewed in Algorithm \ref{algo:MMWU} is presented in \cite{Kale07}, and its performance is evaluated in the following lemma. 
For readers' convenience, we provide a proof of the lemma in Appendix \ref{Proof_Lemma2}; see also \cite[Theorem 3.2]  {Arora_Hazan_Kale12}


\begin{algorithm}[!t]
 \caption{Matrix Multiplicative Weights Update algorithm}
 \label{algo:MMWU}
 \begin{algorithmic}
    \STATE
    Fix $\varepsilon \in(0, \frac{1}{2}),$ and initialize the weight matrix $G(1)=I$.
    \STATE
    For each step $l=1,2,..., L$:
        \STATE 1. ${\cal H}$-player selects mixed state as follows: $F(l) = \frac{G(l)}{\mathrm{Tr}(G(l))}$
        \STATE
        2. ${\cal X}$-player selects pure strategy $x_l \in {\cal X}$.
        \STATE 3. For the cost matrix $M(x_l)$, update the weight matrix:
        \makebox[\linewidth][c]{ $\displaystyle
        G(l+1) = \exp\bigg(-\varepsilon \sum_{l' = 1}^{l} M(x_{l'})\bigg) $}
 \end{algorithmic}
\end{algorithm}

\begin{lemma}(MMWU algorithm performance)  
\label{lem:MMWU_Algorithm}
Fix an arbitrary $\varepsilon\in (0, \frac{1}{2})$.
The mixed states $F(1),F(2),\ldots,F(L)$ which ${\cal H}$-player obtains
by Algorithm \ref{algo:MMWU} satisfies that for every state $\ket{\psi}\in \mathcal{H}$, 
\begin{align*}
        &(1-\varepsilon) \sum_{l=1}^L\bigg(\mathrm{Tr} \Big(F(l)M(x_l)\Big) \bigg) 
        \\&\quad\leq 
        \sum_{l=1}^L \bra{\psi} M(x_l) \ket{\psi} + \frac{\ln{\dim (\mathcal{H})}}{\varepsilon}.
\end{align*}
\end{lemma}


\section{Quantum hypergraph soft covering}
In this section, as a tool to analyze C-Q channel resolvability, we consider the quantum hypergraph soft covering \cite{Ahlswede_Winter02}. A quantum hypergraph soft covering $(\mathcal{V,E})$ consists of a finite-dimensional space $\mathcal{V}$ and a collection of edges $\mathcal{E}$. 
Each edge $E_x\in\mathcal{E}$ is Hermitian matrix  and satisfies $0\preceq E_x \preceq \mathbf{1}_{\mathcal{V}}$ for the identity matrix $\mathbf{1}_{\mathcal{V}}$ on space $\mathcal{V}$. 
The problem of hypergraph soft covering $(\mathcal{V,E})$ is to simulate Hermitian matrix $ E_p = \sum_{E_x\in\mathcal{E}} p(x) E_x,$
by making codebook $\mathcal{C} = \{x_1, x_2, ..., x_L\}$ such that $E_{\mathcal{C}} = \frac{1}{L} \sum_{l=1}^L E_{x_l}.$ 

We will make assumptions about Hermitian matrices $E_x$ and $E_p$ such that for an arbitrary positive constant $\eta,\tau$ and identity matrix $\mathbf{1}_\mathcal{V}$,
\begin{align}\label{eq:assumption_E_x}
    E_x \preceq \eta \mathbf{1}_{\mathcal{V}} \quad \forall E_x \in \mathcal{E},
\end{align}
\begin{align}\label{eq:assumption_Tr_rho_pi}
    \mathrm{Tr}(E_p) \geq 1 - \tau.
\end{align}
These assumptions will be satisfied in Section \ref{sec:Classical-Quantum_channel_resolvability}.

By defining the cost matrix $M(x)$ and $\mathcal{X}$-player's pure strategies $x_1,x_2, \ldots, x_L$ appropriately, we will run the MMWU algorithm, and we obtain a codebook $\mathcal{C}$ from $\mathcal{X}$-player's strategies. Under the assumption \eqref{eq:assumption_E_x} and \eqref{eq:assumption_Tr_rho_pi}, we solves quantum hypergraph soft covering in Theorem \ref{the:Single-shot_for_non-regular_MMWU}. 

\begin{theorem}(Quantum hypergraph soft covering) 
    \label{the:Single-shot_for_non-regular_MMWU}
  For arbitrary positive constants $\varepsilon, \tau\in(0,\frac{1}{2}),  \tau_0\in(0, \frac{1}{2}-\tau), \eta$, the Hermitian matrices $E_p: \mathrm{Tr}(E_p)\geq 1-\tau$ and $E_x \preceq \eta \mathbf{1}_\mathcal{V} \;\forall E_x\in\mathcal{E} $,
    the codebook ${\mathcal{C}}$ obtained by the MMWU algorithm satisfies
    \begin{align}\label{eq:the:Single-shot_for_non-regular_MMWU}
        d_\mathrm{tr}(E_p, E_{\mathcal{C}}) \leq  3\varepsilon + 3\tau+ \frac{7}{2}\tau_0 + \sqrt{2\varepsilon +\tau + \tau_0 },
    \end{align}
    as long as the size $L$ satisfies that 
    \begin{equation} \label{eq:the:Single-shot_for_non-regular_MMWU_L_size}
        L  \geq  \frac{\eta \dim(\mathcal{V}) \ln{\dim(\mathcal{V}) }}{\varepsilon^2 \tau_0}.
    \end{equation}
\end{theorem}
We divide the proof of Theorem \ref{the:Single-shot_for_non-regular_MMWU} into two subsections "before applying assumptions for edges" and "under assumptions for edges".

\subsection{Before applying assumptions for edges}
First, we consider quantum hypergraph soft covering $(\mathcal{V,E})$ before applying  assumptions for edges $E_x\in\mathcal{E}$.
To simulate $E_p$ by constructing the codebook $\mathcal{C}$ from MMWU algorithm, we define the cost matrix judiciously as follows.
\begin{definition}(Cost matrix)
    \label{def:updated_cost_matrix_MMWU}
    Given the Hermitian matrices $E_p$ and $E_x\in \mathcal{E}$, we define 
    \begin{align}\label{eq:def:d_max}
        D_\mathrm{max} := \max_{x} \ln {\lambda_\mathrm{max}\Big(E_p^{-\frac{1}{2}}E_x E_p^{-\frac{1}{2}}\Big)},
    \end{align}
    where $\lambda_\mathrm{max} (A)$ is a largest eigenvalue of Hermitian matrix $A$.
    Then, we define the cost matrix as 
    $$M(x) :=E_p^{-\frac{1}{2}}E_x E_p^{-\frac{1}{2}} e^{-D_\mathrm{max}}.$$
\end{definition}
$D_\mathrm{max}$ bounds the cost matrix as $0\preceq M(x)\preceq \mathbf{1}_{\mathcal{V}}$ for $\mathcal{X}$-player's all strategies $x\in {\cal X}$. 
We run Algorithm \ref{algo:MMWU} for subspace $\mathrm{supp}(E_p)$ instead of $\mathcal{V}$; then ${\cal H}$-player's mixed state $F(l)$ is on subspace $\mathrm{supp}(E_p)$.
We next define ${\cal X}$-player's pure strategy at each round $l$. 
\begin{definition}(${\cal X}$-player's pure strategy) 
    \label{def:X_strategy_MMWU}
    At each round $1\leq l\leq L$, ${\cal H}$-player selects a mixed state $F(l)$. 
    For cost matrix $M(x)$, ${\cal X}$-player selects a pure strategy $x_l\in{\cal X}$ \footnote{If the $x$ that has the maximum value cannot be uniquely determined, choose an arbitrary one.} such that 
    $$x_l := \underset{x}{\mathrm{argmax}}\; \mathrm{Tr}\big(F(l)M(x)\big). $$
\end{definition}

Note that such $\mathcal{X}$-player's pure strategy $x_l$ satisfies the following lemma.

\begin{lemma}
    \label{lem:X_strategy_MMWU}
    At each round $l$, for ${\cal H}$-player's mixed state $F(l)$, ${\cal X}$-player's pure strategy satisfies
    \begin{equation}
        \label{eq:sec4:X_strategy_MMWU}
        \mathrm{Tr}\big(F(l)M(x_l)\big) \geq e^{-D_\mathrm{max}}.
    \end{equation}
\end{lemma}
\begin{proof}
Because ${\cal H}$-player's mixed state is on subspace $\mathrm{supp}(E_p)$, 
we get $F(l) \{E_p>0\} = F(l).$
Thus, we get
\begin{equation*}
    \begin{split}
        &\max_x \mathrm{Tr}\Big(F(l)E_p^{-\frac{1}{2}} E_x E_p^{-\frac{1}{2}} \Big) \\&\qquad\geq \mathbb{E}_{X\sim p}\Big[\mathrm{Tr}\Big(F(l)E_p^{-\frac{1}{2}} E_XE_p^{-\frac{1}{2}} \Big)\Big]
        \\&\qquad=
        \mathrm{Tr}\bigg(F(l) \{E_p > 0\} \bigg) 
        =\mathrm{Tr}(F(l))=1,
    \end{split}
\end{equation*}
where the last equality follows from the fact that $F(l)$ is a density matrix.  
Multiplying $e^{-D_{\mathrm{max}}}$, we get \eqref{eq:sec4:X_strategy_MMWU}. 
\end{proof}

Applying 
Lemma \ref{lem:MMWU_Algorithm} 
to cost matrix definition in Definition \ref{def:updated_cost_matrix_MMWU} 
and $\mathcal{X}$-player's pure strategy in Definition \ref{def:X_strategy_MMWU},  
we get the following 
codebook $\mathcal{C}$ from ${\cal X}$-player's pure strategies.

\begin{lemma}
\label{lem:MMWU_algorithm_peformance_quantum}
  Given the Hermitian matrices $E_p$ and $E_x\in \mathcal{E}$, and arbitrary small $\varepsilon\in(0,\frac{1}{2})$, the codebook $\mathcal{C}$ obtained by the MMWU algorithm satisfies
\begin{equation}\label{eq:sec4:channel_Resolvperfor_}
E_p -E_{\mathcal{C}} \preceq 2\varepsilon  E_p,  \end{equation}
as long as the size $L$ satisfies that 
\begin{equation}\label{eq:sec4:assumption_MMWU_L}
    L \geq \frac{e^{D_\mathrm{max}}\ln{\dim({\cal V})}}{\varepsilon^2}.
\end{equation}
\end{lemma}
\begin{proof}
    First, we derive the equation such that for every state $ \ket{\psi} \in {\cal V},$
    \begin{equation}\label{eq:sec4:4_MMWU}
        \bra{\psi}E_p^{-\frac{1}{2}}E_{\mathcal{C}} E_p^{-\frac{1}{2}} - (1-2\varepsilon) \{ E_p>0\} \ket{\psi} \geq 0,
    \end{equation}
    with three cases of $\ket{\psi}$: $ \ket{\psi_1}\in\mathrm{supp}(E_p)$, $ \ket{\psi_0} \in \mathrm{ker}(E_p)$ and super position of $\ket{\psi_1}, \ket{\psi_0}.$

    Case $\ket{\psi_1}\in\mathrm{supp}(E_p )$: 
    Note that $\dim \Big(\mathrm{supp}(E_p )\Big)\leq \dim (\mathcal{V})$. 
    Applying Lemma \ref{lem:MMWU_Algorithm} on subspace $\mathrm{supp}(E_p)$ for
    Definition \ref{def:updated_cost_matrix_MMWU} and \ref{def:X_strategy_MMWU} and Lemma \ref{lem:X_strategy_MMWU},
    we get for every state $\ket{\psi_1} \in \mathrm{supp}(E_p),$
    \begin{align}
        \label{eq:sec4:4_MMWU_23}
        &\sum_{l=1}^L \bra{\psi_1} M(x_l) \ket{\psi_1}
        \geq (1-\varepsilon)L e^{-D_\mathrm{max}} - \frac{\ln{\dim (\mathcal{V})}}{\varepsilon}.
    \end{align}
    From Definition \ref{def:updated_cost_matrix_MMWU}, we rewrite the left side of equation as
    \begin{align*}
        &\sum_{l=1}^L \bra{\psi_1}E_p^{-\frac{1}{2}}E_{x_l} E_p^{-\frac{1}{2}}\ket{\psi_1} e^{-D_\mathrm{max}} \notag
        \\&\qquad = L \bra{\psi_1}E_p^{-\frac{1}{2}} E_{\mathcal{C}} E_p^{-\frac{1}{2}}\ket{\psi_1} e^{-D_\mathrm{max}}.
    \end{align*}
    Then multiplying $\frac{e^{D_\mathrm{max}}}{L}$ to \eqref{eq:sec4:4_MMWU_23}, we get 
    \begin{align}\label{eq:sec4:4}
    \bra{\psi_1}E_p^{-\frac{1}{2}}E_{\mathcal{C}} {E_p }^{-\frac{1}{2}}\ket{\psi_1} &\geq 1-\varepsilon - \frac{e^{D_\mathrm{max}} \ln{\dim (\mathcal{V})}}{L\varepsilon} \notag
        \\& \geq 1-2\varepsilon,
    \end{align}
    where the last inequality is correct if the size $L$ satisfies \eqref{eq:sec4:assumption_MMWU_L}.
    Note that $\bra{\psi_1} \{E_p > 0\}\ket{\psi_1}= 1$ for every state $\ket{\psi_1}\in \mathrm{supp}(E_p ).$ After multiplying 
    $\bra{\psi_1} \{E_p > 0\}\ket{\psi_1}$ for right side and rearranging \eqref{eq:sec4:4}, \eqref{eq:sec4:4_MMWU} holds.

    Case $\ket{\psi_0}\in \mathrm{ker}(E_p)$: 
    Note that for every state $\ket{\psi_0}\in \mathrm{ker}(E_p)$, we get
    $\{E_p>0\}\ket{\psi_0} = \mathbf{0}.$
    Additionally, $ E_p^{-\frac{1}{2}}  \ket{\psi_0}  = \mathbf{0}$ because of the generalized inverse.   
    \eqref{eq:sec4:4_MMWU} holds.

    Case super position of $\ket{\psi_1}, \ket{\psi_0}$: 
    Note that the case $\ket{\psi_0}$ teaches us $E_p^{-\frac{1}{2}}\ket{\psi_0} = \{E_p>0\} \ket{\psi_0} = \mathbf{0}.$ 
    Thus, the case $\ket{\psi_1} \in \mathrm{supp}(E_p) $ only remains and \eqref{eq:sec4:4_MMWU} holds.

    We get \eqref{eq:sec4:4_MMWU}. From semi-definite matrix definition, \eqref{eq:sec4:4_MMWU} equals that
    $$E_p^{-\frac{1}{2}}E_{\mathcal{C}} E_p^{-\frac{1}{2}} \succeq (1-2\varepsilon) \{E_p>0\}. $$
    After multiplying both sides $ E_p^{-\frac{1}{2}}$, 
    we get 
    $ E_{\mathcal{C}} \succeq  (1-2\varepsilon) E_p$
    as long as the size $L$ satisfies \eqref{eq:sec4:assumption_MMWU_L}.
\end{proof}

\subsection{Under assumptions for edges}
We consider quantum hypergraph soft covering $(\mathcal{V,E})$ under assumption \eqref{eq:assumption_E_x} and \eqref{eq:assumption_Tr_rho_pi} for edges $E_x\in\mathcal{E}$. 
Given an projector $\Pi $ and a Hermitian matrix $A$, we denote the pinching map as 
    $$A^\Pi  := \Pi \; A \Pi.$$ 
    Fix $\tau_0$ as an arbitrary positive constant. Diagonalize $E_p$ such as $E_p = \sum_{j}r_j\pi_j,$  and construct
    \begin{align*}
        \Pi_1 := \sum_{j: r_j > \frac{\tau_0}{\dim(\mathcal{V})}} \pi_j, \qquad \Pi_0 := \sum_{j: r_j \leq \frac{\tau_0}{\dim(\mathcal{V})}} \pi_j.
    \end{align*}
    Note that the projectors $\Pi_0, \Pi_1$ commutes to $E_p$.   
    From definition of $\Pi_0, \Pi_1$ and assumption \eqref{eq:assumption_Tr_rho_pi}, we get 
    \begin{align}\label{eq:Tr_rho_pi_0}
        \mathrm{Tr}\Big(E_p^{\Pi_0}  \Big) =  \mathrm{Tr}\Big(\Pi_0 E_p {\Pi_0}  \Big) \leq \frac{\tau_0\mathrm{Tr}(\Pi_0)}{\dim(\mathcal{V})}\leq \tau_0,
    \end{align}
    \begin{align}\label{eq:Tr_rho_pi_1}
        \mathrm{Tr}\Big(E_p^{\Pi_1}  \Big) = \mathrm{Tr}(E_p) - \mathrm{Tr}\Big(E_p^{\Pi_0}  \Big) \geq  1- \tau - \tau_0,
    \end{align}
    \begin{align}\label{eq:rho_pi_1_property}
        (\Pi_1 E_p \Pi_1)^{-1} \preceq \frac{ \dim(\mathcal{V})} {\tau_0}\Pi_1.
    \end{align}
    We apply Definition \ref{def:updated_cost_matrix_MMWU} to $E_x^{\Pi_1}, E_p^{\Pi_1}$ instead of $E_x, E_p$. 
    Then, $D_\mathrm{max}$ is bounded by the following lemma.

\begin{lemma}(Bound of $D_\mathrm{max}$)\label{lem:d_max_bound}
    Given Hermitian matrices $E_p$ and $E_x \preceq \eta \mathbf{1}_{\mathcal{V}}$ $E_x\in\mathcal{E}$, 
    we get bound of $D_\mathrm{max}$ such that
    \begin{equation}
        \label{eq:bound_of_D_max_0}
        D_\mathrm{max} \leq \ln \frac{\eta \dim(\mathcal{V})}{\tau_0}.
    \end{equation} 
\end{lemma}
\begin{proof}
    By multiplying $\Pi_1$ and $\Big({E_p^{\Pi_1}}\Big)^{-\frac{1}{2}}$ to assumption \eqref{eq:assumption_E_x} from both sides in this order, we get
    \begin{align*}
       \Big({E_p^{\Pi_1}}\Big)^{-\frac{1}{2}} E_x^{\Pi_1}\Big({E_p^{\Pi_1}}\Big)^{-\frac{1}{2}} &\preceq \eta \Big({E_p^{\Pi_1}}\Big)^{-\frac{1}{2}} \Pi_1 \Big({E_p^{\Pi_1}}\Big)^{-\frac{1}{2}} \notag
       \\&= \eta \Pi_1 \Big({E_p^{\Pi_1}}\Big)^{-1} \Pi_1   \preceq   \frac{\eta\dim (\mathcal{V})}{\tau_0} \Pi_1
    \end{align*}
    where the equality follows from $\Pi_1$ commuting to $E_p^{\Pi_1}$ and the last inequality follows from \eqref{eq:rho_pi_1_property}.
    Thus, substituting 
    above equation to \eqref{eq:def:d_max}, we get \eqref{eq:bound_of_D_max_0}. 
\end{proof}

Under the assumption \eqref{eq:assumption_E_x} and \eqref{eq:assumption_Tr_rho_pi}, 
applying Lemma \ref{lem:MMWU_algorithm_peformance_quantum} to $E_x^{\Pi_1}, E_p^{\Pi_1}$ instead of $E_x, E_p$, 
the MMWU algorithm solves quantum hypergraph soft covering as follows.

\textit{Proof of Theorem \ref{the:Single-shot_for_non-regular_MMWU}.}
    By applying \eqref{eq:sec4:assumption_MMWU_L} to \eqref{eq:bound_of_D_max_0}, we get size of codebook $|\mathcal{C}|= L$ in \eqref{eq:the:Single-shot_for_non-regular_MMWU_L_size}. 
    By applying the triangle inequality twice, we measure the trace distance and get
    \begin{align}\label{eq:the:Single-shot_for_non-regular_MMWU_1}
    d_\mathrm{tr}(E_p, E_{\mathcal{C}}) 
    &\leq \frac{1}{2}\| E_p - E_p^{\Pi_1}  \|_1 + \frac{1}{2}\|E_{\mathcal{C}}^{\Pi_1} - E_{\mathcal{C}}\|_1 \notag
    \\& \quad + \frac{1}{2} \|E_p^{\Pi_1} - E_{\mathcal{C}}^{\Pi_1}\|_1 
    \end{align}
    The first term of \eqref{eq:the:Single-shot_for_non-regular_MMWU_1} equals 
    \begin{equation}\label{eq:the:Single-shot_for_non-regular_MMWU_a1}
        \|E_p - E_p^{\Pi_1} \|_1  = \|E_p{\Pi_0} \|_1
        =\mathrm{Tr}(E_p^{\Pi_0}) \leq \tau_0
    \end{equation}
    where the first equality follows from projector $\Pi_1$ commuting to $E_p$ and the inequality follows from \eqref{eq:Tr_rho_pi_0}. 
    We shall bound the second term of \eqref{eq:the:Single-shot_for_non-regular_MMWU_1}. Note that 
    \begin{align} \label{eq:the:Single-shot_for_non-regular_MMWU_b0}
        \mathrm{Tr}\Big(E_{\mathcal{C}}^{\Pi_1}\Big) &= \mathrm{Tr}\Big( E_p^{\Pi_1}\Big)- \bigg(\mathrm{Tr}\Big( E_p^{\Pi_1}\Big) - \mathrm{Tr}\Big(E_{\mathcal{C}}^{\Pi_1}\Big)\bigg) \notag
        \\& \geq \mathrm{Tr}\Big( E_p^{\Pi_1}\Big) - 2\varepsilon \geq 1 - 2\varepsilon- \tau - \tau_0 ,
    \end{align}
    where the second inequality follows from \eqref{eq:sec4:channel_Resolvperfor_} and the last inequality follows from \eqref{eq:Tr_rho_pi_1}. Using gentle measurement lemma (e.g., see \cite[Lemma 9.4.2]{Wilde_2018}) for $\mathrm{Tr}\Big(E_{\mathcal{C}}^{\Pi_1}\Big) \geq  1- 2\varepsilon- \tau- \tau_0 $,  
    we get
    \begin{align}\label{eq:the:Single-shot_for_non-regular_MMWU_b1}
        \|E_{\mathcal{C}} - E_{\mathcal{C}}^{\Pi_1} \|_1
        \leq 2\sqrt{2\varepsilon+\tau + \tau_0  }.
    \end{align}
    To measure the last term of \eqref{eq:the:Single-shot_for_non-regular_MMWU_1}, we prepare density matrices 
    $$\sigma_{E_p^{\Pi_1}} = \frac{E_p^{\Pi_1}}{\mathrm{Tr} (E_p^{\Pi_1})}, \qquad \sigma_{E_{\mathcal{C}}^{\Pi_1}}= \frac{E_{\mathcal{C}}^{\Pi_1}}{\mathrm{Tr}(E_{\mathcal{C}}^{\Pi_1})}.$$
    These density matrices have a good relation such that 
    \begin{align}\label{eq:the:Single-shot_for_non-regular_MMWU_c0}
        \sigma_{E_p^{\Pi_1}} - \sigma_{E_{\mathcal{C}}^{\Pi_1}}
        &\preceq \frac{E_p^{\Pi_1}}{\mathrm{Tr}(E_p^{\Pi_1})} - E_{\mathcal{C}}^{\Pi_1} \notag
        \\&\preceq  \frac{E_p^{\Pi_1}}{1 - \tau - \tau_0} - E_{\mathcal{C}}^{\Pi_1}\notag
        \\&\preceq (1 + 2 (\tau + \tau_0))E_p^{\Pi_1} -E_{\mathcal{C}}^{\Pi_1} \notag
        \\&\preceq 2(\varepsilon + \tau + \tau_0) E_p^{\Pi_1},
    \end{align}
    where the first inequality follows from $\sigma_{E_{\mathcal{C}}^{\Pi_1}}\succeq E_{\mathcal{C}}^{\Pi_1}$, 
    the second inequality follows from \eqref{eq:Tr_rho_pi_1}, 
    the third inequality follows that for $z = \tau + \tau_0 \in (0, \frac{1}{2})$, we get 
    $\frac{1}{1- z} \leq 2z + 1,$
    and the last inequality follows from \eqref{eq:sec4:channel_Resolvperfor_}.
    By applying the triangle inequality twice, we get 
    \begin{align}\label{eq:the:Single-shot_for_non-regular_MMWU_c1}
        \|E_p^{\Pi_1} -E_{\mathcal{C}}^{\Pi_1}\|_1 &\leq  \|\sigma_{E_p^{\Pi_1}} - \sigma_{E_{\mathcal{C}}^{\Pi_1}}\|_1 + \|\sigma_{E_p^{\Pi_1}} - E_p^{\Pi_1} \|_1 \notag 
        \\& \quad + \|\sigma_{E_{\mathcal{C}}^{\Pi_1}}-E_{\mathcal{C}}^{\Pi_1}\|_1 .
    \end{align}
    Because of density matrix $\sigma_{E_p^{\Pi_1}}, \sigma_{E_{\mathcal{C}}^{\Pi_1}}$, we get
    \begin{equation}\label{eq:the:Single-shot_for_non-regular_MMWU_c2}
        \|\sigma_{E_p^{\Pi_1}} - \sigma_{E_{\mathcal{C}}^{\Pi_1}}\|_1 = 2\mathrm{Tr}\Big(\sigma_{E_p^{\Pi_1}} - \sigma_{E_{\mathcal{C}}^{\Pi_1}}\Big)_+ 
        \leq 4(\varepsilon + \tau + \tau_0),
    \end{equation}
    where the inequality follows from \eqref{eq:the:Single-shot_for_non-regular_MMWU_c0}.
    Calculating the other terms in \eqref{eq:the:Single-shot_for_non-regular_MMWU_c1}, we get
    \begin{align}\label{eq:the:Single-shot_for_non-regular_MMWU_c3}
        \|\sigma_{E_p^{\Pi_1}} - E_p^{\Pi_1}\|_1 &= \bigg(\frac{1}{\mathrm{Tr}(E_p^{\Pi_1})} -1 \bigg) \| E_p^{\Pi_1}\|_1 \notag
        \\&= 1- \mathrm{Tr}(E_p^{\Pi_1})
         \leq \tau + \tau_0,
    \end{align}
    where the second equality follows form $\|E_p^{\Pi_1}\|_1 = \mathrm{Tr}(E_p^{\Pi_1})$  
    and the inequality follows from  \eqref{eq:Tr_rho_pi_1}. We also get
    \begin{align}\label{eq:the:Single-shot_for_non-regular_MMWU_c4}
        \|\sigma_{E_{\mathcal{C}}^{\Pi_1}}- E_{\mathcal{C}}^{\Pi_1}\|_1 &= 1 - \mathrm{Tr}\Big( E_{\mathcal{C}}^{\Pi_1}\Big)
        \leq 2\varepsilon + \tau +\tau_0,
    \end{align}
    where 
    the inequality follows from \eqref{eq:the:Single-shot_for_non-regular_MMWU_b0}.
    Substituting \eqref{eq:the:Single-shot_for_non-regular_MMWU_a1}, \eqref{eq:the:Single-shot_for_non-regular_MMWU_b1}, \eqref{eq:the:Single-shot_for_non-regular_MMWU_c1}, \eqref{eq:the:Single-shot_for_non-regular_MMWU_c2}, 
    \eqref{eq:the:Single-shot_for_non-regular_MMWU_c3} and  \eqref{eq:the:Single-shot_for_non-regular_MMWU_c4} to \eqref{eq:the:Single-shot_for_non-regular_MMWU_1}, we get \eqref{eq:the:Single-shot_for_non-regular_MMWU}. \qed



\section{Classical-Quantum channel resolvability}\label{sec:Classical-Quantum_channel_resolvability}
In this section, we prove C-Q channel resolvability for block length $n$, using fixed type $T$ on $\mathcal{X}$; see also \cite[Section 2]{Csiszar-Korner11}. 
Let $\mathcal{X}_T^n$ be a set of all sequences $x^n$ in $\mathcal{X}^n$ that have type $T$ and $p_n^T$ be an input distribution on $\mathcal{X}^n_T$. We shall simulate the output state $W_{p_n^T}$ 
by applying the quantum hypergraph soft covering.  
First, we introduce the typical projector and the conditional typical projector; see  \cite{Winter99}. 

\begin{lemma}(Typical projector)\label{lem:Typical_projector}
    For $\rho$ on $\mathcal{H}$ and an arbitrary positive constant $\alpha$, the orthogonal subspace projector $\Pi_{\rho, \alpha}^n$ commuting to $\rho^{\otimes n}$ exists such that
    \begin{equation}\label{eq:lem:Typical_projector_2_Tr_Pi}
        \mathrm{Tr}(\Pi_{\rho,\alpha}^n)  \leq \exp \Big(nH(\rho) + \frac{1}{e}\dim (\mathcal{H})\alpha \sqrt{n}\Big),
    \end{equation}
    where $H(\rho)$ is von Neumann Entropy.  
\end{lemma}
\begin{lemma}(Conditional typical projector)\label{lem:Conditional_Typical_projector}
    For $x^n \in \mathcal{X}_T^n$, stationary memoryless C-Q channel $W$, and an arbitrary positive constant $\alpha$, the orthogonal subspace projector $\Pi_{W, \alpha}^n(x^n)$ commuting to $W_{x^n}$ exists such that
    \begin{equation}\label{eq:lem:Conditional_Typical_projector_1_Tr_rho_Pi}
        \mathrm{Tr}(W_{x^n}\Pi_{W,\alpha}^n(x^n)) \geq 1- \frac{\dim (\mathcal{H}) |\mathcal{X}|}{\alpha^2},
    \end{equation}
    \begin{align} \label{eq:lem:Conditional_Typical_projector2_Tr_Pi}
        &\Pi_{W,\alpha}^n (x^n) W_{x^n} \Pi_{W,\alpha}^n (x^n) \leq \notag
        \\& \quad \Pi_{W,\alpha}^n(x^n) \exp \Big(-nH(W|T) + \frac{1}{e}\dim (\mathcal{H}) |\mathcal{X}|\alpha \sqrt{n}\Big),
    \end{align}
    \begin{equation}
    \label{eq:Conditional_Typical_projector_8_weak_law_of_large_number}
    \mathrm{Tr}\Big(W_{x^n} \Pi_{W_T,\alpha\sqrt{\mathcal{|X|}}}^n\Big)\geq 1 - \frac{\dim (\mathcal{H})\mathcal{|X|}}{\alpha^2}, 
\end{equation}
    where $H(W|T)$ is conditional von Neumann Entropy.
\end{lemma}
Using the projectors, we define each edge $E_x\in\mathcal{E}$ as
\begin{equation}\label{eq:edge_Q}
    Q_{x^n} = \Pi_{W_T, \alpha\sqrt{|\mathcal{X}|}}^n\Pi_{W,\alpha}^n(x^n) W_{x^n} \Pi_{W,\alpha}^n(x^n)\Pi_{W_T, \alpha\sqrt{|\mathcal{X}|}}^n.
\end{equation}
From \eqref{eq:lem:Conditional_Typical_projector_1_Tr_rho_Pi} and  \eqref{eq:Conditional_Typical_projector_8_weak_law_of_large_number}, $\mathrm{Tr}(Q_{x^n})$ is bounded by next Lemma. For readers' convenience, we provide a proof of the lemma in Appendix \ref{Proof:Bound_of_Q_Trace}.
\begin{lemma}(Bound of $Q_{x^n}$ trace)\label{lem:Tr_Q_bound}
    For $x^n \in \mathcal{X}_T^n,$ the Hermitian matrix $Q_{x^n}$ is bounded as 
    \begin{equation}\label{eq_Tr_Q}
        \mathrm{Tr}(Q_{x^n}) \geq 1 - \frac{2\dim(\mathcal{H})\mathcal{|X|}}{\alpha^2},
    \end{equation}
\end{lemma}

Now, we have recipes to prove C-Q channel resolvability for fixed type $T$.  
We apply the quantum hypergraph soft covering in Theorem \ref{the:Single-shot_for_non-regular_MMWU} with the range of $\Pi_{W_T,\alpha\sqrt{|\mathcal{X}|}}^n$ as vertex space and edges $Q_{x^n} $ for $x^n \in \mathcal{X}_T^n.$

\begin{theorem}(C-Q channel resolvability for fixed type)\label{the:quantum_channel_resolvability_for_fixed_tyoe}
For a fixed type $T$, an input distribution $p_n^T$, stationary memoryless C-Q channel $W$, and arbitrary positive constants $\varepsilon,\tau\in(0,\frac{1}{2}), \tau_0 \in(0, \frac{1}{2}- \tau)$, 
the codebook $\mathcal{C}_n^T$ obtained by the MMWU algorithm satisfies 
\begin{align}\label{eq:lem:Applying_quantum_hypergraph_0}
    &d_{\mathrm{tr}}(W_{p_n^T}, W_{\mathcal{C}_n^T}) \notag \leq 3\varepsilon + 3\tau+ \frac{7}{2}\tau_0 +\sqrt{2\varepsilon +\tau + \tau_0 }
    \\&\quad\qquad \qquad \qquad  + 2\sqrt{\tau} + \sqrt{2\tau}
\end{align}    
as long as the size $L_n^T$ satisfies that 
\begin{align}\label{eq:lem:Applying_quantum_hypergraph_01}
    &L_n^T \geq \notag\exp{\Big(nI(T,W) + \frac{\sqrt{2}}{e\sqrt{\tau}}\dim (\mathcal{H})^{\frac{3}{2}}(|\mathcal{X}| + |\mathcal{X}|^{\frac{3}{2}}) \sqrt{n} \Big)}
    \\& \qquad \quad \times \frac{ n \ln{\dim(\mathcal{H}) }}{\varepsilon^2 \tau_0} 
\end{align}
where $I(T,W) = H(W_T) - H(W|T)$ is the Holevo information. 
\end{theorem}

\begin{proof}
     We verify that the edge $Q_{x^n}$ in \eqref{eq:edge_Q} satisfies the assumption \eqref{eq:assumption_E_x} and \eqref{eq:assumption_Tr_rho_pi}. 
     Clearly, $ \Pi_{W,\alpha}^n(x^n) \preceq I^{\otimes n}$ for identity matrix $I^{\otimes n}$. 
     Fix $\eta$ as \begin{align}\label{eq:lem:Applying_quantum_hypergraph_1}
     \eta = \exp\big(-nH(W|T) + \frac{1}{e}\dim (\mathcal{H}) \mathcal{|X|}\alpha \sqrt{n} \big),
     \end{align}
     and by multiplying $\Pi_{W_T, \alpha\sqrt{|\mathcal{X}|}}^n =\mathbf{1}_{\mathcal{V}}$ to \eqref{eq:lem:Conditional_Typical_projector2_Tr_Pi} from both sides, 
     we get assumption \eqref{eq:assumption_E_x}.
     Next, we verify $\mathrm{Tr}(Q_{p_n^T}) \geq 1 - \tau$. From \eqref{eq_Tr_Q}, we get
     \begin{equation}\label{eq:lem:Applying_quantum_hypergraph_2}
         \mathrm{Tr}(Q_{p_n^T}) = \sum_{x^n \in \mathcal{X}_T^n}p_{n}^T(x^n) \mathrm{Tr}(Q_{x^n}) \geq 1- \frac{2\dim(\mathcal{H})\mathcal{{|X|}}}{\alpha^2}. 
     \end{equation}
     Fix $\alpha = \sqrt{\frac{2\dim(\mathcal{H})\mathcal{{|X|}}}{\tau}},$ we get assumption \eqref{eq:assumption_Tr_rho_pi}. 
     Next, we shall get an upper bound of $\dim(\mathcal{V})$. Because the vertex space is the range of $\Pi_{W_T, \alpha\sqrt{|\mathcal{X}|}}^n$, 
     \eqref{eq:lem:Typical_projector_2_Tr_Pi} suggests that 
     \begin{align}\label{eq:lem:Applying_quantum_hypergraph_3}
         \dim (\mathcal{V})
         \leq \exp\Big(n H(W_T) + \frac{1}{e}\dim(\mathcal{H}) \sqrt{\mathcal{|X|}} \alpha  \sqrt{n}  \Big),
     \end{align}
     Trivially, we also get $\dim (\mathcal{V}) \leq n \dim (\mathcal{H}).$
    Applying Theorem \ref{the:Single-shot_for_non-regular_MMWU} to \eqref{eq:lem:Applying_quantum_hypergraph_1}, \eqref{eq:lem:Applying_quantum_hypergraph_2} and \eqref{eq:lem:Applying_quantum_hypergraph_3}, 
     we get \eqref{eq:lem:Applying_quantum_hypergraph_01}. 
     
     Here, using gentle measurement lemma (e.g., see \cite[Lemma 9.4.2]{Wilde_2018}), for  \eqref{eq:lem:Conditional_Typical_projector_1_Tr_rho_Pi} \eqref{eq_Tr_Q} with $\alpha = \sqrt{\frac{2\dim(\mathcal{H})\mathcal{{|X|}}}{\tau}} $, we get $\| Q_{x^n} - W_{x^n} \|_1  \leq 2\sqrt{\tau }+\sqrt{2\tau}$ and 
    \begin{align*}
        \| Q_{p_n^T} - W_{p_n^T} \|_1 &\leq \sum_{x^n\in \mathcal{X}_T^n} p_{n}^T(x^n) \|Q_{x^n} - W_{x^n}\|_1 \\&\leq 2\sqrt{\tau}+\sqrt{2\tau}.
    \end{align*}
   Similarly, $\|Q_{\mathcal{C}_n^T} - W_{\mathcal{C}_n^T} \|_1 \leq \sum_{l=1}^{L_n^T}\|Q_{x_l^n} - W_{x_l^n}  \|_1 \leq 2\sqrt{\tau}+\sqrt{2\tau}. $ 
    Thus, using the triangle inequality in \eqref{eq:the:Single-shot_for_non-regular_MMWU}, we get \eqref{eq:lem:Applying_quantum_hypergraph_0}.
    
\end{proof}

By applying Theorem \ref{the:quantum_channel_resolvability_for_fixed_tyoe} to each type, we can prove C-Q channel resolvability for a general input distribution $p_n$. 
Since the argument is the same as classical case, we only explain the outline in the following; see \cite[Section 6]{takahashi2025channelresolvabilityusingmultiplicative} for detail. 
Let $\mathcal{T}$ be the set of all types and $p_\mathcal{T}$ be the type distribution such that $p_\mathcal{T}(T) = p_n(\mathcal{X}_T^n)$. Then, $p_n$ decomposes into input distributions $p^T_n$ on type classes $\mathcal{X}_T^n$ with weight $p_\mathcal{T}$. 
We can simulate the type distribution $p_\mathcal{T}$ by using classical MWU algorithm in \cite{takahashi2025channelresolvabilityusingmultiplicative} with randomness of size in $\ln n$ order. For each type $T$, we select a codebook according to Theorem \ref{the:quantum_channel_resolvability_for_fixed_tyoe}. 
By appropriately choosing the parameters, we can prove the C-Q channel resolvability for general input distribution. 


\begin{theorem}(C-Q channel resolvability for general input distribution) \label{the:quantum_channel_resolvability_}
For arbitrary positive $\xi,\kappa$,  for sufficiently large $n$, and for every general input distribution $p_n$,   by using the MMWU algorithm, 
we can construct a codebook $\mathcal{C}_n$ such that $d_{\mathrm{tr}}(W_{p_n}, W_{\mathcal{C}_n}) \leq\xi,$ 
as long as the size $L_n$ satisfies that $L_n \geq \exp\big(n(C(W)+\kappa)\big)$
where $C(W)$ is the Holevo capacity. 
\end{theorem}

For an i.i.d. input distribution, since types concentrates around the input distribution, we can show the following theorem from Theorem \ref{the:quantum_channel_resolvability_for_fixed_tyoe}. Since the proof is the same as classical case, we omit the proof; see \cite{takahashi2025channelresolvabilityusingmultiplicative} for details. 
\begin{theorem}(For i.i.d. input distribution)
    For arbitrary positive $\xi,\kappa$,  for sufficiently large $n$, and for an i.i.d. input distribution $p_n =p^n$, by using the MMWU algorithm, 
    we can construct a codebook $\mathcal{C}_n$ such that $d_{\mathrm{tr}}(W_{p_n}, W_{\mathcal{C}_n}) \leq\xi,$ 
        as long as the size $L_n$ satisfies that $L_n \geq \exp\big(n(I(p,W)+\kappa)\big)$
    where $I(p,W)$ is the Holevo information. 
\end{theorem}

\section*{Acknowledgment}
This work was supported in part by the Japan Society for the Promotion of Science (JSPS) KAKENHI under Grant 23H00468 and 23K17455 and by JST, CRONOS, Japan Grant Number JPMJCS25N5.

\newpage

\appendix
\subsection{Proof of Lemma \ref{lem:MMWU_Algorithm}}
\label{Proof_Lemma2}
    We shall calculate the upper bound and lower bound of $\mathrm{Tr}(G(L+1))$. We calculate the upper bound first. Here, we have 
    \begin{align}\label{eq:appe:1}
            &\mathrm{Tr}\big(G(L+1)\big) \notag \\&\quad= \mathrm{Tr}\bigg(\exp{\bigg(-\varepsilon \sum_{l=1}^L M(x_{l}) \bigg)}\bigg) \notag
            \\&\quad \leq \mathrm{Tr}\bigg(\exp \bigg(-\varepsilon \sum_{l=1}^{L-1} M(x_{l}) \bigg) \exp\big(-\varepsilon M(x_L)\big)\bigg) \notag
            \\&\quad= \mathrm{Tr}\Big(G(L) \exp{\big(-\varepsilon M(x_L)\big)}\Big),
    \end{align}
    where the inequality follows from the Golden-Thompson inequality \cite{Golden65,Thompson_65} that for given Hermitian matrices $A,B$, we get $\mathrm{Tr}(e^{A+B})\leq \mathrm{Tr}(e^A e^B).$
    Define $\varepsilon_1:= 1-e^{-\varepsilon}.$ 
    Note that, for $0 \preceq A \preceq I$ and $0\leq \varepsilon$, we get
    \begin{equation}\label{eq:appe:nontrivial_1}
        \exp(-\varepsilon A)\preceq I - \varepsilon_1 A.
    \end{equation}
    We provide a proof of inequality \eqref{eq:appe:nontrivial_1} with Lemma \ref{Matrix_function} in Appendix \ref{claim:A}.   
    From \eqref{eq:appe:nontrivial_1}, we get
    \begin{equation*}
        \exp{\big(-\varepsilon M(x_L)\big)} \preceq I-\varepsilon_1 M(x_L).
    \end{equation*}
    Thus, we rewrite \eqref{eq:appe:1} as 
    \begin{equation*}\label{eq:appe:2}
        \begin{split}
            \mathrm{Tr}\Big(G(L+1)\Big) \leq  \mathrm{Tr}\Big(G(L) -  \varepsilon_1 G(L) M(x_L)\Big).
        \end{split}
    \end{equation*}
    We continue to calculate and get
    \begin{align}\label{eq:appe:3}
            \mathrm{Tr}\Big(G(L+1)\Big) &\leq  \mathrm{Tr}\Big(G(L)\Big) 
            - \mathrm{Tr}\Big(\varepsilon_1  G(L) M(x_L)\Big) \notag
            \\&= \mathrm{Tr}\big(G(L)\big) \Big(1- \varepsilon_1\mathrm{Tr}\Big(F(L) M(x_L)) \Big)\notag
            \\& \leq \mathrm{Tr}\Big(G(L)\Big) \exp \bigg(- \varepsilon_1\mathrm{Tr}\Big(F(L) M(x_L)\Big) \bigg),
    \end{align}
    where the equality follows from $$G(l) = {F(l)\mathrm{Tr}(G(l))},$$ in Algorithm \ref{algo:MMWU} and the last inequality follows from $1+z \leq e^z$ for any real number $z$. Calculate \eqref{eq:appe:3} recursively, the upper bound is 
    \begin{align} \label{eq:appe:upperbound}
            &\mathrm{Tr}\Big(G(L+1)\Big) \notag
            \\& \quad \leq \mathrm{Tr}\Big(G(L) \Big)\exp \bigg(-\varepsilon_1 \mathrm{Tr}\Big(F(L)M(x_L)\Big) \bigg) \notag 
            \\& \quad\leq \mathrm{Tr}\Big(G(L-1) \Big)
            \exp \bigg(-\varepsilon_1 \sum_{l=L-1}^L \mathrm{Tr}\Big( F(l) M(x_l)\Big) \bigg) \notag 
            \\&\quad \leq   \cdots \notag
            \\& \quad\leq \mathrm{Tr}\big(G(1) \big)\exp \bigg(-\varepsilon_1 \sum_{l=1}^L\mathrm{Tr}\Big( F(l) M(x_l)\Big) \bigg) \notag
            \\&\quad = \dim (\mathcal{H}) \exp \bigg(-\varepsilon_1 \sum_{l=1}^L\mathrm{Tr}\Big( F(l) M(x_l)\Big) \bigg),
    \end{align}
    where the last equality follows from $\mathrm{Tr}\big(G(1)\big) = \mathrm{Tr}(I) = \dim (\mathcal{H}).$
    On the other hand, the lower bound is, for every $i \in \{1,2,..., \dim (\mathcal{H})\}$,
    \begin{align}\label{eq:appe:lowerbound}
            \mathrm{Tr}\Big( G(L+1)\Big) &= \mathrm{Tr}\bigg(\exp{\bigg(-\varepsilon \sum_{l=1}^L M(x_l)\bigg)}\bigg) \notag
            \\ & \geq \exp{\bigg(-\varepsilon \lambda_i \bigg(\sum_{l=1}^L M(x_l)\bigg)\bigg)},
    \end{align}
  where the inequality follows that $$\mathrm{Tr}(\exp(A))= \sum_{i=1}^{\dim (\mathcal{H})}\exp(\lambda_i(A))\geq \exp(\lambda_i(A)),$$ for $i$-th eigenvalue $\lambda_i(A)$ of Hermitian matrix $A$.
    Combining the upper bound in \eqref{eq:appe:upperbound} and the lower bound in \eqref{eq:appe:lowerbound}, we get
    \begin{equation*}
        \begin{split}
            &\exp\bigg(-\varepsilon \lambda_i\bigg(\sum_{l=1}^L M(x_l)\bigg) \bigg) 
            \\&\quad\leq 
        \dim (\mathcal{H}) \exp \bigg( -\varepsilon_1 \sum_{l=1}^L \mathrm{Tr}\Big(M(x_l)\Big)\bigg),
        \end{split} 
    \end{equation*}
    for every $i\in\{1,2,...,\dim (\mathcal{H})\}$. Simplify the above equation with
    \begin{equation}
    \label{eq:appe:nontrivial_2}
        \varepsilon_1 = 1 - e^{-\varepsilon} \geq \varepsilon(1-\varepsilon),
    \end{equation}
    which holds for $\frac{1}{2}\geq \varepsilon \geq 0$; we provide a proof of inequality \eqref{eq:appe:nontrivial_2} in Appendix \ref{claim:B}. And taking the natural logarithm and dividing by $-\varepsilon \neq 0$, we get
    \begin{equation*}
        \begin{split}
            &\lambda_i\bigg(\sum_{l = 1}^L M(x_l)\bigg) 
             \\&\quad\geq (1-\varepsilon) \sum_{l=1}^L\bigg(\mathrm{Tr} \Big(F(l)M(x_l)\Big) \bigg) - \frac{\ln\dim (\mathcal{H})}{\varepsilon}.
        \end{split}
    \end{equation*}
    Now, for every unit vector $\ket{\psi} \in\mathcal{H}$, we get $$\min_i \lambda_i \bigg(\sum_{l=1}^LM(x_l)\bigg) \leq   \sum_{l=1}^L\bra{\psi}M(x_l)\ket{\psi}.$$
\qed

\subsubsection{Proof of nontrivial inequality \eqref{eq:appe:nontrivial_1}}
    \label{claim:A}
    We use useful lemma; see also \cite[Chapter 3, Lemma 1]{Kale07}.
    \begin{lemma}(Matrix and function)\label{Matrix_function}
    Let $f, g: \mathbb{R\rightarrow R},$ and suppose that the inequality $f(x)\geq g(x)$ holds for $x\in D$ for some $D\subset \mathbb{R}$. Then for every symmetric matrix $A$ all of whose eigenvalues lie in $D$, we have $f(A)\succeq g(A)$.
    \end{lemma}
    \begin{proof}
        Let $A = UDU^T$ be the diagonalization of $A$. Then $f(A)-g(A) = U(f(D)-g(D))U^T,$ and since all eigenvalues of $A$ are in $D$, $f(D) -g(D)$ is diagonal matrix with non-negative matrix. Thus, $f(A)-g(A)$ is a positive semi-definite matrix, and hence $f(A)\succeq g(A)$.
    \end{proof}
    We define 
    $$f(\varepsilon) := \exp{(-\varepsilon z)} - (1 - (1-\exp{(-\varepsilon)}) z).$$
    Taking the derivatives of $f(\varepsilon)$ with respect to $\varepsilon$, we get
    $$f'(\varepsilon) = - z\exp{(-\varepsilon z)} + z\exp{(-\varepsilon)} \leq 0,$$
    where the inequality follows from $\exp{(-\varepsilon z)} \geq \exp{(-\varepsilon)}$ for $0\leq z \leq 1$ and $\varepsilon\geq0.$
    We get
    $$f(\varepsilon)\leq f(0) = 0.$$
    Thus, we get the inequality such that
    $$\exp{(-\varepsilon z)} \leq  1 - (1-\exp{(-\varepsilon)}) z.$$
    Applying Lemma \ref{Matrix_function} to this equation, we get $$\exp{(-\varepsilon A)} \leq  1 - (1-\exp{(-\varepsilon)}) A,$$
    for matrix $A$ which satisfies $0\preceq A\preceq I$.
\qed

\subsubsection{Proof of nontrivial inequality \eqref{eq:appe:nontrivial_2}}
\label{claim:B}
We define $g(\varepsilon)$  as 
$$g(\varepsilon) := 1- \exp{(-\varepsilon)} - \varepsilon(1-\varepsilon).$$
Taking the first and second derivatives of $g(\varepsilon)$ with respect to $\varepsilon$, we get
$$g'(\varepsilon) = \exp{(-\varepsilon)} -1 + 2\varepsilon,$$
and
$$g''(\varepsilon) = -\exp{(-\varepsilon)} + 2.$$
For $\frac{1}{2}\geq \varepsilon \geq 0$, we get $g''(\varepsilon)>0.$
Thus, we get 
$$g'(\varepsilon) \geq g'(0) = 0,$$
and
$$g(\varepsilon)\geq g(0) = 0.$$
Then, we get the inequality such that
$$ 1 - \exp{(-\varepsilon)} \geq \varepsilon(1-\varepsilon).$$
\qed

\subsection{Proof of Lemma \ref{lem:Tr_Q_bound}}
\label{Proof:Bound_of_Q_Trace}
    We calculate and get
    \begin{align}\label{eq_Tr_Q_0}
        \mathrm{Tr}(Q_{x^n}) &= \mathrm{Tr}\Big(\Pi_{W,\alpha}^n(x^n) W_{x^n} \Pi_{W,\alpha}^n(x^n)\Pi_{W_T, \alpha\sqrt{|\mathcal{X}|}}^n\Big) \notag
        \\&= \mathrm{Tr}\Big( W_{x^n} \Pi_{W,\alpha}^n(x^n)\Pi_{W_T, \alpha\sqrt{|\mathcal{X}|}}^n\Big) \notag
        \\&= \mathrm{Tr}\Big(W_{x^n}\Pi_{W_T, \alpha\sqrt{|\mathcal{X}|}}^n\Big) \notag
        \\& \quad- \mathrm{Tr}\Big( W_{x^n} \Big(I^{\otimes n} - \Pi_{W,\alpha}^n(x^n) \Big)\Pi_{W_T, \alpha\sqrt{|\mathcal{X}|}}^n\Big)
    \end{align}
    where the second equality follows that $\Pi_{W,\alpha}^n(x^n)$ commutes to $W_{x^n}$. Because of \eqref{eq:Conditional_Typical_projector_8_weak_law_of_large_number}, we get
    \begin{align}\label{eq_Tr_Q_1}
        \mathrm{Tr}\Big(W_{x^n}\Pi_{W_T, \alpha\sqrt{|\mathcal{X}|}}^n\Big) \geq 1- \frac{\dim(\mathcal{H}) \mathcal{|X|}}{\alpha^2}.
    \end{align}
    Note that for semi-definite matrix $A$ and projector $\Pi$ the equation holds $\mathrm{Tr}(A \Pi) \leq \mathrm{Tr}(A)$. 
    Thus, we get 
    \begin{align}\label{eq_Tr_Q_2}
        &\mathrm{Tr}\Big( W_{x^n} \Big(I^{\otimes n} - \Pi_{W,\alpha}^n(x^n) \Big)\Pi_{W_T, \alpha\sqrt{|\mathcal{X}|}}^n\Big) \notag
        \\&\qquad = \mathrm{Tr}\Big(\Big(I^{\otimes n} - \Pi_{W,\alpha}^n(x^n) \Big) W_{x^n} \Big(I^{\otimes n} - \Pi_{W,\alpha}^n(x^n) \Big) \notag
        \\& \qquad \qquad \qquad \Pi_{W_T, \alpha\sqrt{|\mathcal{X}|}}^n\Big) \notag
        \\&\qquad \leq \mathrm{Tr}\Big(\Big(I^{\otimes n} - \Pi_{W,\alpha}^n(x^n) \Big) W_{x^n} \Big(I^{\otimes n} - \Pi_{W,\alpha}^n(x^n) \Big)\Big) \notag 
        \\&\qquad= 1 - \mathrm{Tr}(W_{x^n} \Pi_{W,\alpha}^n(x^n)) \leq \frac{\dim (\mathcal{H}) \mathcal{|X|}}{\alpha^2}
    \end{align}
    where the first equality follows that $\Pi_{W,\alpha}^n(x^n)$ commutes to $W_{x^n}$ and the last inequality follows from \eqref{eq:lem:Conditional_Typical_projector_1_Tr_rho_Pi}. 
    By substituting \eqref{eq_Tr_Q_1} and \eqref{eq_Tr_Q_2} to \eqref{eq_Tr_Q_0}, we get \eqref{eq_Tr_Q}. \qed

\bibliographystyle{plain}
\bibliography{sample}

\end{document}